\documentclass[11pt]{article}
\usepackage{graphicx}
\usepackage{multirow}
\usepackage[T1]{fontenc}
\usepackage[sc]{mathpazo}
\usepackage{amsmath}
\usepackage{amssymb}
\usepackage{enumerate}
\usepackage{amsthm}% theorems and proofs
\usepackage{amsfonts,mathrsfs}
\usepackage[style]{fncychap}
\usepackage{graphicx} % For including graphics N.B. pdftex graphics driver
\usepackage{geometry} % allows easy specifying of the page layout
\usepackage{eepic}
\usepackage{ifthen}
\newboolean{ElectronicVersion}
\setboolean{ElectronicVersion}{true} % CHANGE THIS AS REQUIRED

\usepackage{hyperref}
\hypersetup{pdfpagemode=UseNone}

\newcommand{\tinyspace}{\mspace{1mu}}

\newcommand{\abs}[1]{\left\lvert\tinyspace #1 \tinyspace\right\rvert}

\renewcommand{\t}{{\scriptscriptstyle\mathsf{T}}}

\newcommand{\setft}[1]{\mathrm{#1}}

\newcommand{\density}[1]{\setft{D}\left(#1\right)}
\newcommand{\unitary}[1]{\setft{U}\left(#1\right)}

\def\real{\mathbb{R}}

\newenvironment{mylist}[1]{\begin{list}{}{
    \setlength{\leftmargin}{#1}
    \setlength{\rightmargin}{0mm}
    \setlength{\labelsep}{2mm}
    \setlength{\labelwidth}{8mm}
    \setlength{\itemsep}{0mm}}}
    {\end{list}}

\def\ot{\otimes}

\newcommand{\defeq}{\stackrel{\smash{\textnormal{\tiny def}}}{=}}

\newcommand{\Pa}[1]{\left(#1\right)}

\newcommand{\Br}[1]{\left[#1\right]}
\newcommand{\set}[1]{\{#1\}}
\newcommand{\Set}[1]{\left\{#1\right\}}

\DeclareMathOperator{\trace}{Tr}

\newcommand{\Ptr}[2]{\trace_{#1}\Pa{#2}}

\newcommand{\Tr}[1]{\Ptr{}{#1}}

\def\cH{\mathcal{H}}

\def\cS{\mathcal{S}}
\def\cU{\mathcal{U}}

\def\rH{\mathrm{H}}

\def\rS{\mathrm{S}}

\newtheorem{thrm}{Theorem}[section]

\newtheorem{prop}[thrm]{Proposition}

\newtheorem{conj}[thrm]{Conjecture}
\theoremstyle{definition}

\numberwithin{equation}{section}

\newcounter{questionnumber}

\begin{document}

%============================================================================================================%
\title{\Large Quantum marginal inequalities and the conjectured entropic inequalities}
%============================================================================================================%

\author{Lin Zhang\footnote{E-mail: godyalin@163.com}\ , Hongjin He\\
  {\it\small Institute of Mathematics, Hangzhou Dianzi University, Hangzhou 310018, PR~China}\\
  Yuan-hong Tao\\
  {\it \small College of Science,Yanbian University, Yanji 133002, PR~China}}
\date{}
\maketitle
\maketitle \mbox{}\hrule\mbox\\
\begin{abstract}

A conjecture -- \emph{the modified super-additivity inequality} of
relative entropy -- was proposed in \cite{Zhang2012}: There exist
three unitary operators $U_A\in \unitary{\cH_A},U_B\in
\unitary{\cH_B}$, and $U_{AB}\in \unitary{\cH_A\ot \cH_B}$ such that
$$
\rS(U_{AB}\rho_{AB}U^\dagger_{AB}||\sigma_{AB}) \geqslant \rS(U_A\rho_AU^\dagger_A||\sigma_A) + \rS(U_B\rho_BU^\dagger_B||\sigma_B),
$$
where the reference state $\sigma$ is required to be full-ranked. A
numerical study on the conjectured inequality is conducted in this
note. The results obtained indicate that the modified
super-additivity inequality of relative entropy seems to hold for
all qubit pairs.

\end{abstract}
\maketitle \mbox{}\hrule\mbox

%=============================================================================%
\section{Introduction}
%=============================================================================%

Rau derived in \cite{Rau2010} a wrong monotonicity property of
relative entropy, and based on this false inequality he obtained the
\emph{super-additivity inequality} -- a much stronger monotonicity
-- of relative entropy:
$$
\rS(\rho_{AB}||\sigma_{AB})\geqslant \rS(\rho_A||\sigma_A) +
\rS(\rho_B||\sigma_B),
$$
where $\rho_{AB}$ and $\sigma_{AB}$ are bipartite states on
$\cH_A\ot\cH_B$. A simple counterexample \cite{Zhang2012} was provided to show that the above inequality is not correct. Moreover, it is \emph{conjectured} that, there exist three unitary operators $U_A\in \unitary{\cH_A},U_B\in \unitary{\cH_B}$, and $U_{AB}\in \unitary{\cH_A\ot \cH_B}$ such that
\begin{eqnarray}\label{eq:conjectured-ineq}
\rS(U_{AB}\rho_{AB}U^\dagger_{AB}||\sigma_{AB}) \geqslant \rS(U_A\rho_AU^\dagger_A||\sigma_A) + \rS(U_B\rho_BU^\dagger_B||\sigma_B),
\end{eqnarray}
where the reference state $\sigma$ is required to be full-ranked.

A numerical study on the conjectured inequality is conducted in this
note. The results obtained indicate that the modified
super-additivity inequality of relative entropy seems to hold for
all qubit pairs. An attempt is made to give some potential
applications in quantum information theory.

Before proceeding, we need to fix some notations. If the column vectors
$$
p = [p_1,\ldots,p_d]^\t\in \real^d,\quad q = [q_1,\ldots,q_d]^\t\in \real^d
$$
are two probability distributions, the \emph{Shannon entropy} of $p$ is defined by
$$
\rH(p) \defeq - \sum_{i=1}^d p_i\log_2 p_i,
$$
where $x\log_2 x := 0$ if $x=0$, and the \emph{relative entropy} of $p$ and $q$ is defined by
$$
\rH(p||q) \defeq \sum^d_{i=1} p_i(\log_2p_i - \log_2 q_i).
$$
Let $\density{\cH_d}$ denote the set of all density matrices $\rho$ on a $d$-dimensional Hilbert space $\cH$. The
\emph{von Neumann entropy} $\rS(\rho)$ of $\rho$ is defined by
$$
\rS(\rho) \defeq - \Tr{\rho\log\rho}.
$$
In fact, this definition can be equivalently described as follows: if we denote the vector consisting of eigenvalues of $\rho$ by $\lambda(\rho) = [\lambda_1(\rho),\ldots,\lambda_d(\rho)]^\t$, then we have
$$
\rS(\rho) = \rH(\lambda(\rho)) = \rH(\lambda^\downarrow(\rho)),
$$
where we write $\lambda^\downarrow(\rho)$ for a vector with components being the same as $\lambda(\rho)$ and arranged in non-increasing order, i.e.
$$
\lambda^\downarrow(\rho) = [\lambda^\downarrow_1(\rho),\ldots,\lambda^\downarrow_d(\rho)]^\t\quad(\lambda^\downarrow_1(\rho)\geqslant\cdots\geqslant\lambda^\downarrow_d(\rho)).
$$
However,  $\lambda^{\uparrow}(\rho)$ stands for the vector with
eigenvalues of $\rho$ arranged in increasing order. The
\emph{relative entropy} of two mixed states $\rho$ and $\sigma$ is
defined by
$$
\rS(\rho||\sigma) \defeq \left\{\begin{array}{ll}
                             \Tr{\rho(\log\rho -
\log\sigma)}, & \text{if}\ \mathrm{supp}(\rho) \subseteq
\mathrm{supp}(\sigma), \\
                             +\infty, & \text{otherwise}.
                           \end{array}
\right.
$$

%=============================================================================%
\section{Technical lemmas}
%=============================================================================%

The so-called \emph{quantum marginal problem}, i.e. the existence of
mixed states $\rho_{AB}$ two (or multi-) component system $\cH_{AB}
= \cH_A\ot\cH_B$ with reduced density matrices $\rho_A,\rho_B$ and
given spectra $\lambda_{AB},\lambda_A,\lambda_B$, is discussed in
the literature, and a complete solution of this problem in terms of
linear inequalities on the spectra is given in the following
proposition.

\begin{prop}[Klyachko, \cite{Klyachko2006}]
Assume that there is a bipartite system $AB$, described by Hilbert
space $\cH_{AB} = \cH_A\ot\cH_B$. All constraints on spectra
$\lambda^\downarrow(\rho_X) = \lambda^X(X= A, B, AB)$, arranged in
non-increasing order, are given by the following linear
inequalities:
\begin{eqnarray}
\sum_{i=1}^{m} a_i \lambda^A_{\alpha(i)} + \sum_{j=1}^{n} b_j
\lambda^B_{\beta(j)} \leqslant \sum_{k=1}^{mn} (a+b)^\downarrow_k
\lambda^{AB}_{\gamma(k)},
\end{eqnarray}
where $a:a_1\geqslant a_2\geqslant\cdots\geqslant a_m$,
$b:b_1\geqslant b_2\geqslant\cdots\geqslant b_n$ with $\sum_{i=1}^m
a_i = \sum_{j=1}^nb_j=0$ are "test spectra", the spectrum
$(a+b)^\downarrow$ consists of numbers $a_i+b_j$ arranged in
non-increasing order, and $\alpha\in \cS_m,\beta\in\cS_n, \gamma\in
\cS_{mn}$ are permutations subject to a topological condition
$c^\gamma_{\alpha\beta}(a,b)\neq0$, where the meaning of
$c^\gamma_{\alpha\beta}$ can be found in \cite{Klyachko2006}.
\end{prop}

In particular, for the simplest quantum multipartite system, i.e.
two-qubit system, there is a nice solution for the quantum marginal
problem:

\begin{prop}[Bravyi, \cite{Bravyi2004}]\label{lem:constraints}
Mixed two-qubit state $\rho_{AB}$ with spectrum $\lambda_1\geqslant \lambda_2\geqslant\lambda_3\geqslant \lambda_4\geqslant0$ and margins $\rho_A,\rho_B$ exists if and only if minimal eigenvalues $\lambda_A,\lambda_B$ of the margins satisfy inequalities
\begin{eqnarray}
\begin{cases}
\min(\lambda_A,\lambda_B) \geqslant \lambda_3 + \lambda_4,\\
\lambda_A + \lambda_B \geqslant \lambda_2 + \lambda_3 + 2\lambda_4,\\
\abs{\lambda_A - \lambda_B} \leqslant \min(\lambda_1 - \lambda_3, \lambda_2 - \lambda_4).
\end{cases}
\end{eqnarray}
\end{prop}

The following result attempts to give a possibility to corrected version
of superadditivity inequality. Here we give another proof in terms of matrix analysis language.

\begin{prop}[Zhang, \cite{Zhang}]\label{prop:Zhang}
For given two quantum states $\rho,\sigma\in\density{\cH_d}$, where $\sigma$ is invertible, it holds that
\begin{eqnarray}
\min_{U\in\unitary{\cH_d}} \rS(U\rho U^\dagger||\sigma) &=& \rH(\lambda^\downarrow(\rho)||\lambda^\downarrow(\sigma)),\\
\max_{U\in\unitary{\cH_d}} \rS(U\rho U^\dagger||\sigma) &=&
\rH(\lambda^\downarrow(\rho)||\lambda^\uparrow(\sigma)),
\end{eqnarray}
where $\lambda^{\uparrow}_j(\sigma)$ stands for the eigenvalues
arranged in increasing order. $\unitary{\cH_d}$ denotes the set of
all unitary operators on $\cH_d$. Moreover, the set $\set{\rS(U\rho
U^\dagger||\sigma): U\in\unitary{\cH_d}}$ is identical to an
interval:
$$
\Set{\rS(U\rho U^\dagger||\sigma): U\in\unitary{\cH_d}} =
\Br{\rH(\lambda^\downarrow(\rho)||\lambda^\downarrow(\sigma)),
\rH(\lambda^\downarrow(\rho)||\lambda^\uparrow(\sigma))}.
$$
\end{prop}

\begin{proof}
Apparently, the unitary orbit $\cU_\rho$ of $\rho$ is a compact set.
Moreover, every differentiable curve through $\rho$ can be
represented locally as $\exp(tK)\rho \exp(-tK)$ for some
skew-Hermitian $K$, i.e. $K^\dagger = -K$. The derivative of this
curve at $t=0$ is $[K,\rho]:=K\rho-\rho K$ \cite{Bhatia1997}.

Let $f(U) := \rS(U\rho U^\dagger||\sigma)$ be defined over the
unitary group $\unitary{\cH_d}$. Clearly
$$
f(U) = -\rS(\rho) -\Tr{U\rho U^\dagger\log\sigma}.
$$
Since the unitary group $\unitary{\cH_d}$ is a path-connected and
compact space \cite{Baker2003}, it suffices to show that $f(U)$ is a
continuous function.

Let $U_t = \exp(tK)$ for an arbitrary skew-Hermitian $K$. Thus
\begin{eqnarray}
\frac{df(U_t)}{dt} = \Tr{U_t\rho U^\dagger_t[K,\log\sigma]},
\end{eqnarray}
implying
$$
\frac{df(U_t)}{dt}|_{t=0} = \Tr{\rho[K,\log\sigma]},
$$
which means
that $f(U)$ is continuous over $\unitary{\cH_d}$.

Without loss of generality, we assume that $U_0\in \unitary{\cH_d}$
is the extreme point of $f$. Consider an arbitrary differentiable
path $\set{\exp(tK)U_0}$ through $U_0$ in $\unitary{\cH_d}$ for
arbitrary skew-Hermitian $K$, it follows that
\begin{eqnarray*}
\frac{df(\exp(tK)U_0)}{dt}\big|_{t=0} &=&
\Tr{U_0\rho U^\dagger_0[K,\log\sigma]}\\
&=& \Tr{K[\log\sigma, U_0\rho U^\dagger_0]} \\
&=& 0.
\end{eqnarray*}
Thus, by the arbitrariness of $K$, we have $[\log\sigma, U_0\rho U^\dagger_0] = 0$. That is $[\sigma, U_0\rho U^\dagger_0] = 0$.
By the \emph{rearrangement inequality} in mathematics, the desired
conclusion is obtained.
\end{proof}
In fact, partial results in the above proposition has already been
reported in \cite{Zhang2012}. It was employed to study a modified
version of super-additivity inequality of relative entropy.

The above theorem also gives rise to the following inequality:
\begin{eqnarray}
\rH(\lambda^\downarrow(\rho)||\lambda^\downarrow(\sigma))\leqslant
\rS(\rho||\sigma) \leqslant
\rH(\lambda^\downarrow(\rho)||\lambda^\uparrow(\sigma)).
\end{eqnarray}

If we denote $\triangle \rS = \rS(\rho_{AB}||\sigma_{AB}) -
\rS(\rho_A||\sigma_A)- \rS(\rho_B||\sigma_B)$, then we have the
following inequality:
\begin{eqnarray}
\bar\Delta\leqslant \triangle \rS \leqslant\Delta,
\end{eqnarray}
where
\begin{eqnarray}
\bar\Delta
\defeq\rH(\lambda^\downarrow(\rho_{AB})||\lambda^\downarrow(\sigma_{AB}))
- \rH(\lambda^\downarrow(\rho_A)||\lambda^\uparrow(\sigma_A)) -
\rH(\lambda^\downarrow(\rho_B)||\lambda^\uparrow(\sigma_B)).
\end{eqnarray}

In order to study the sign of $\triangle\rS$, we now propose to
study the following four differences:
\begin{eqnarray}
\Delta_{\min} &\defeq& \rH(\lambda^\downarrow(\rho_{AB})||\lambda^\downarrow(\sigma_{AB})) - \rH(\lambda^\downarrow(\rho_A)||\lambda^\downarrow(\sigma_A)) - \rH(\lambda^\downarrow(\rho_B)||\lambda^\downarrow(\sigma_B)),\label{eq:reduction1}\\
\Delta_{\max} &\defeq& \rH(\lambda^\downarrow(\rho_{AB})||\lambda^\uparrow(\sigma_{AB})) - \rH(\lambda^\downarrow(\rho_A)||\lambda^\uparrow(\sigma_A)) - \rH(\lambda^\downarrow(\rho_B)||\lambda^\uparrow(\sigma_B)),\label{eq:reduction2}\\
\Delta_{\mathrm{mix}} &\defeq& \rH(\lambda^\downarrow(\rho_{AB})||\lambda^\uparrow(\sigma_{AB})) - \rH(\lambda^\downarrow(\rho_A)||\lambda^\uparrow(\sigma_A)) - \rH(\lambda^\downarrow(\rho_B)||\lambda^\downarrow(\sigma_B)),\label{eq:reduction3}\\
\Delta &\defeq&
\rH(\lambda^\downarrow(\rho_{AB})||\lambda^\uparrow(\sigma_{AB})) -
 \rH(\lambda^\downarrow(\rho_A)||\lambda^\downarrow(\sigma_A)) -
\rH(\lambda^\downarrow(\rho_B)||\lambda^\downarrow(\sigma_B)).\label{eq:reduction4}
\end{eqnarray}
An observation is made here:
$$
\bar\Delta\leqslant \Delta_{\min},\quad
\bar\Delta\leqslant\Delta_{\max}\leqslant
\Delta_{\mathrm{mix}}\leqslant \Delta.
$$
It can be seen that we can choose suitable qubit pair
$(\rho_{AB},\sigma_{AB})$ to ensure that $\triangle\rS$ can take
arbitrary values in the interval $[\bar\Delta,\Delta]$, which is
guaranteed by Proposition~\ref{prop:Zhang}.

In fact, by Proposition~\ref{prop:Zhang}, if we can show that at
least one of the above-mentioned four quantities is nonnegative,
then our conjectured inequality is correct.

Analytical proof concerning the above inequalities are expected. Proving these seems to be very difficult. Thus we turn to another method -- a numerical study in lower dimensions.

Consider a two-qubit pair $\rho_{AB},\sigma_{AB}$. Let
$\lambda^\downarrow(\rho_{AB}) =
[\lambda_1,\lambda_2,\lambda_3,\lambda_4]$ with
$\lambda_1\geqslant\lambda_2\geqslant\lambda_3\geqslant\lambda_4\geqslant0$
and $\sum_j \lambda_j = 1$; $\lambda^\downarrow(\sigma_{AB}) =
[\mu_1,\mu_2,\mu_3,\mu_4]$ with
$\mu_1\geqslant\mu_2\geqslant\mu_3\geqslant\mu_4>0$ and $\sum_j
\mu_j = 1$. Then the corresponding eigenvalue vectors of their reduced
density matrices, i.e. margins, are $\lambda^\downarrow(\rho_X) =
[1-\lambda_X, \lambda_X]^\t$ with $\lambda_X\in [0,\tfrac12]$.
Similarly, $\lambda^\downarrow(\sigma_X) = [1-\mu_X, \mu_X]^\t$ with
$\mu_X\in (0,\tfrac12]$. Note that $X=A,B$ in the above
formulations.

In what follows, we make a numerical study of each quantity defined
by Eq.~\eqref{eq:reduction1}--Eq.~\eqref{eq:reduction4} under the
constraints \eqref{lem:constraints} for a two-qubit pair $\rho_{AB}$
and $\sigma_{AB}$.

%=============================================================================%
\section{Numerical study}
%=============================================================================%

In this section, we investigate the numerical performance of the
modified superadditivity inequality of the relative entropy to
verify the correctness of our conjecture. Our tests were conducted
using {\sc Matlab} R2010b, and the random data were generated by the
function "\verb"rand"" in {\sc Matlab}.

We test two scenarios with respect to one thousand and one million
groups of random data for each quantity defined by
Eq.~\eqref{eq:reduction1}--Eq.~\eqref{eq:reduction4}. The
corresponding plots are listed in
Fig.~\ref{fig:1}--Fig.~\ref{fig:4}. Obviously, from
Fig.~\ref{fig:1}, we can see that the difference $\Delta_{\min}$
defined by Eq.~\eqref{eq:reduction1} is less than zero in many
cases. Note that in Fig.~\ref{fig:2}, there is only one negative
value of $\Delta_{\max}$ for the one thousand scenario, and a very
small number of points are located below the X-axis for the second
scenario. However, from Fig.~\ref{fig:3} and Fig.~\ref{fig:4}, it is
clear that all the differences $\Delta_{\text{mix}}$ and $\Delta$,
respectively, defined by Eq.~\eqref{eq:reduction3} and
Eq.~\eqref{eq:reduction4} are greater than zero, which supports our
conjecture.

Therefore analytical proof for the following two inequalities are
expected:
\begin{eqnarray}
\rH(\lambda^\downarrow(\rho_{AB})||\lambda^\uparrow(\sigma_{AB}))
&\geqslant&
\rH(\lambda^\downarrow(\rho_A)||\lambda^\uparrow(\sigma_A)) +
\rH(\lambda^\downarrow(\rho_B)||\lambda^\downarrow(\sigma_B)),\label{eq:mix}\\
\rH(\lambda^\downarrow(\rho_{AB})||\lambda^\uparrow(\sigma_{AB}))
&\geqslant&
 \rH(\lambda^\downarrow(\rho_A)||\lambda^\downarrow(\sigma_A)) +
\rH(\lambda^\downarrow(\rho_B)||\lambda^\downarrow(\sigma_B)).\label{eq:pure}
\end{eqnarray}
In fact, if Eq.~\eqref{eq:mix} holds, then Eq.~\eqref{eq:pure} a
\emph{fortiori} holds. Based on these numerical studies, we can make
a bold conjecture:
\begin{conj}
$\rS(U_A\ot U_B\rho_{AB}U^\dagger_A\ot U^\dagger_B||\sigma_{AB})
\geqslant \rS(U_A\rho_AU^\dagger_A||\sigma_A) +
\rS(U_B\rho_BU^\dagger_B||\sigma_B)$ for some unitaries
$U_X\in\unitary{\cH_X}$, where $X=A,B$.
\end{conj}

We give a little remark on the above conjecture. To prove it, we
need to characterize local unitary equivalence between two bipartite
states. We say that $\rho_{AB}$ is local unitary equivalent to
$\rho'_{AB}$ if there exist unitaries $U_X\in\unitary{\cH_X}(X=A,B)$
such that
$$
\rho'_{AB} = (U_A\ot U_B) \rho_{AB}(U_A\ot U_B)^\dagger.
$$
Along with this line, the readers, for instance, can be referred to
\cite{Fei2012}.

%%%============== Old version of tex  ===============
\begin{figure}[htbp]
\centering
\includegraphics[width=0.49\textwidth]{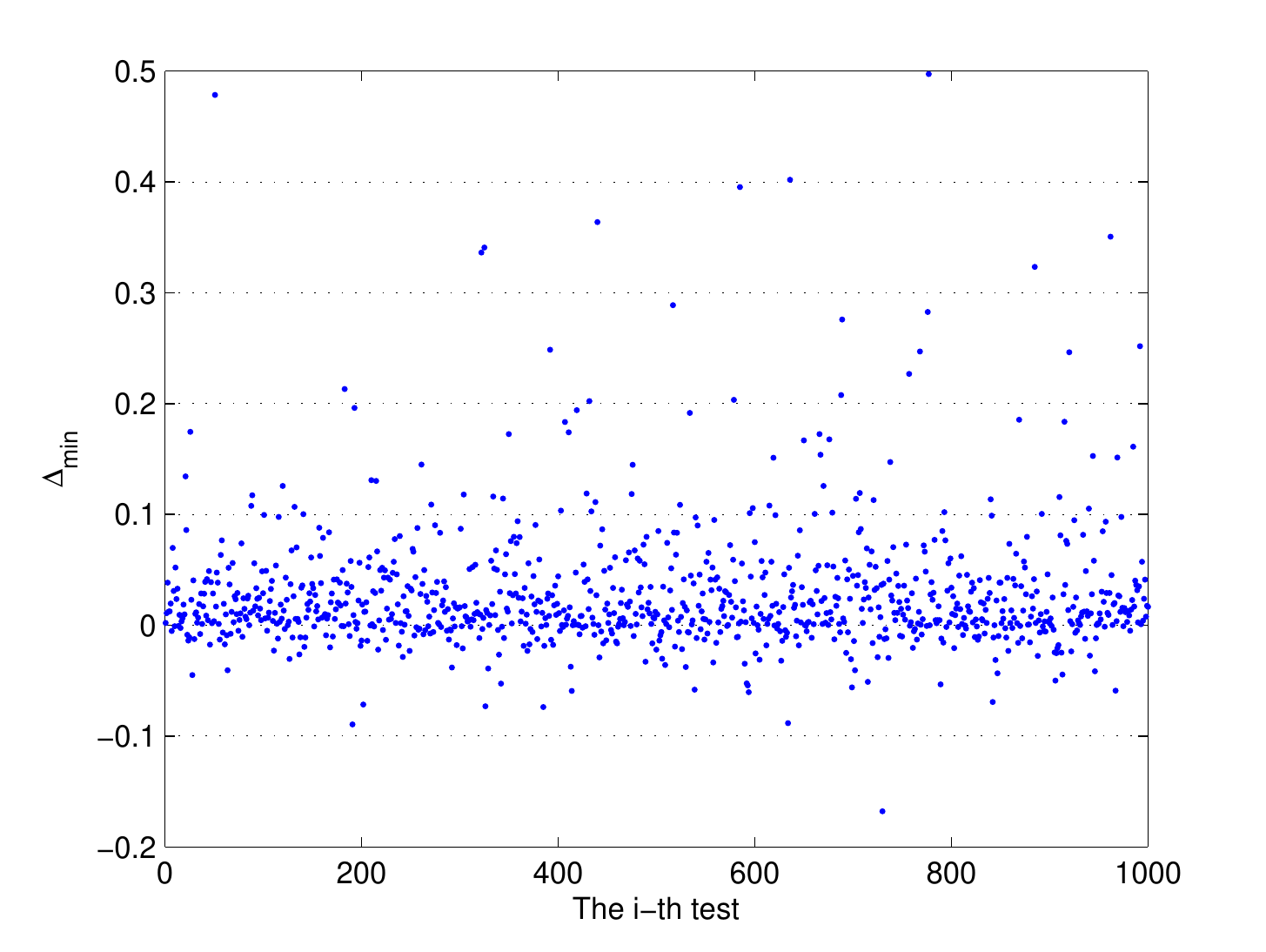}
\includegraphics[width=0.49\textwidth]{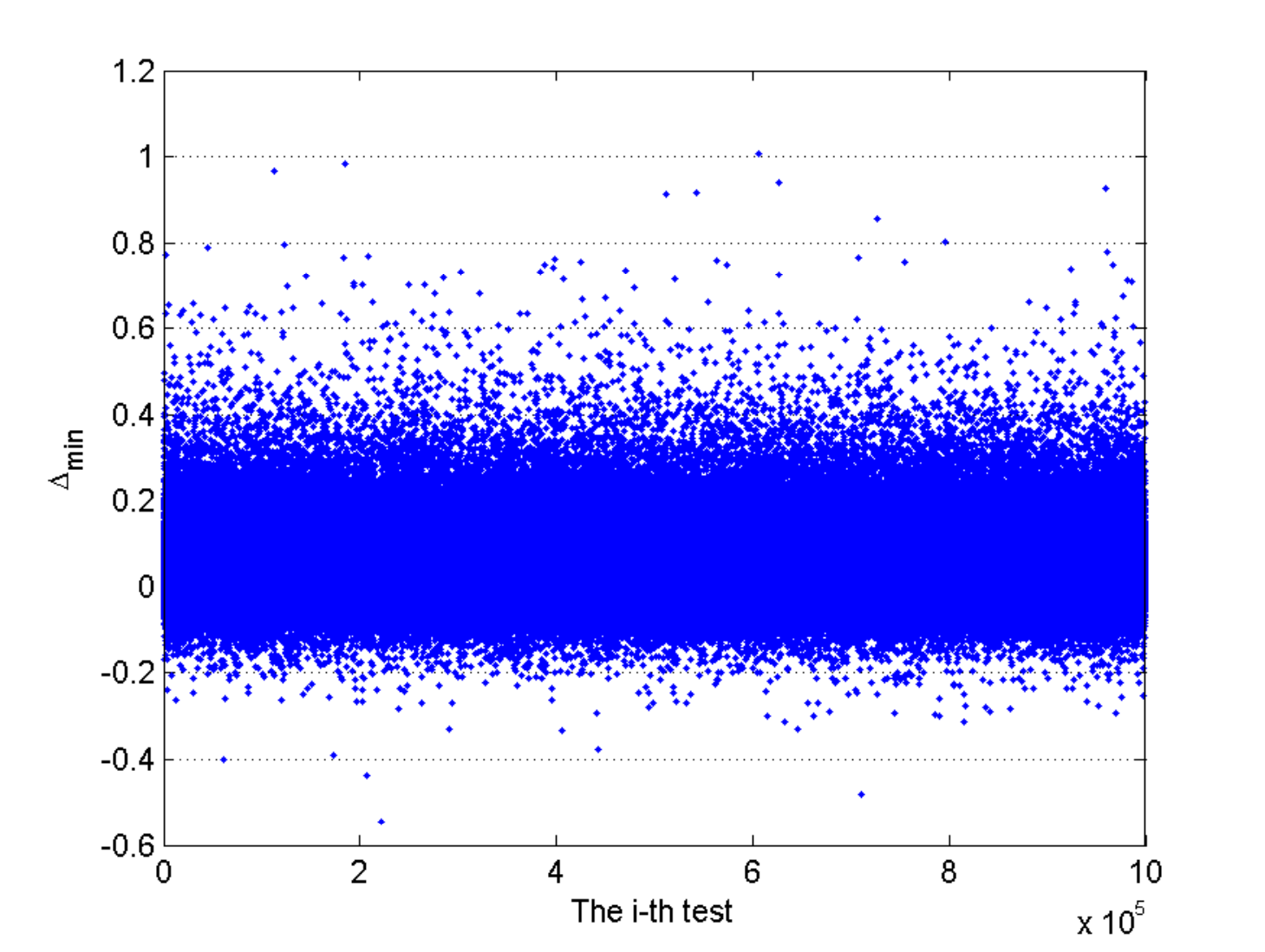}\caption{Two scenarios with respect to one thousand (the left one) and one million (the right one) groups of random data for testing $\Delta_{\min}$.}\label{fig:1}
\end{figure}

\begin{figure}[htbp]
\centering
\includegraphics[width=0.49\textwidth]{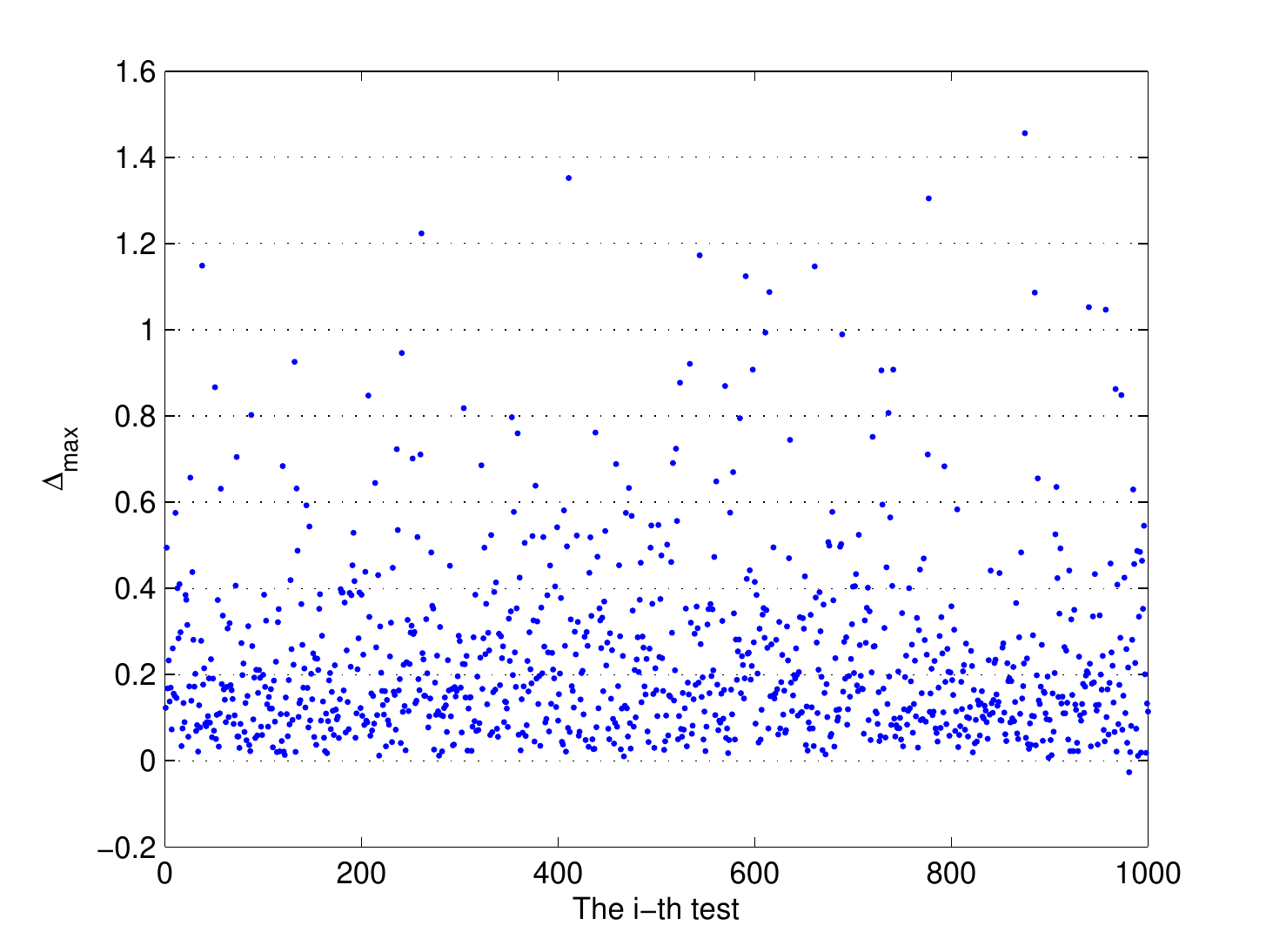}
\includegraphics[width=0.49\textwidth]{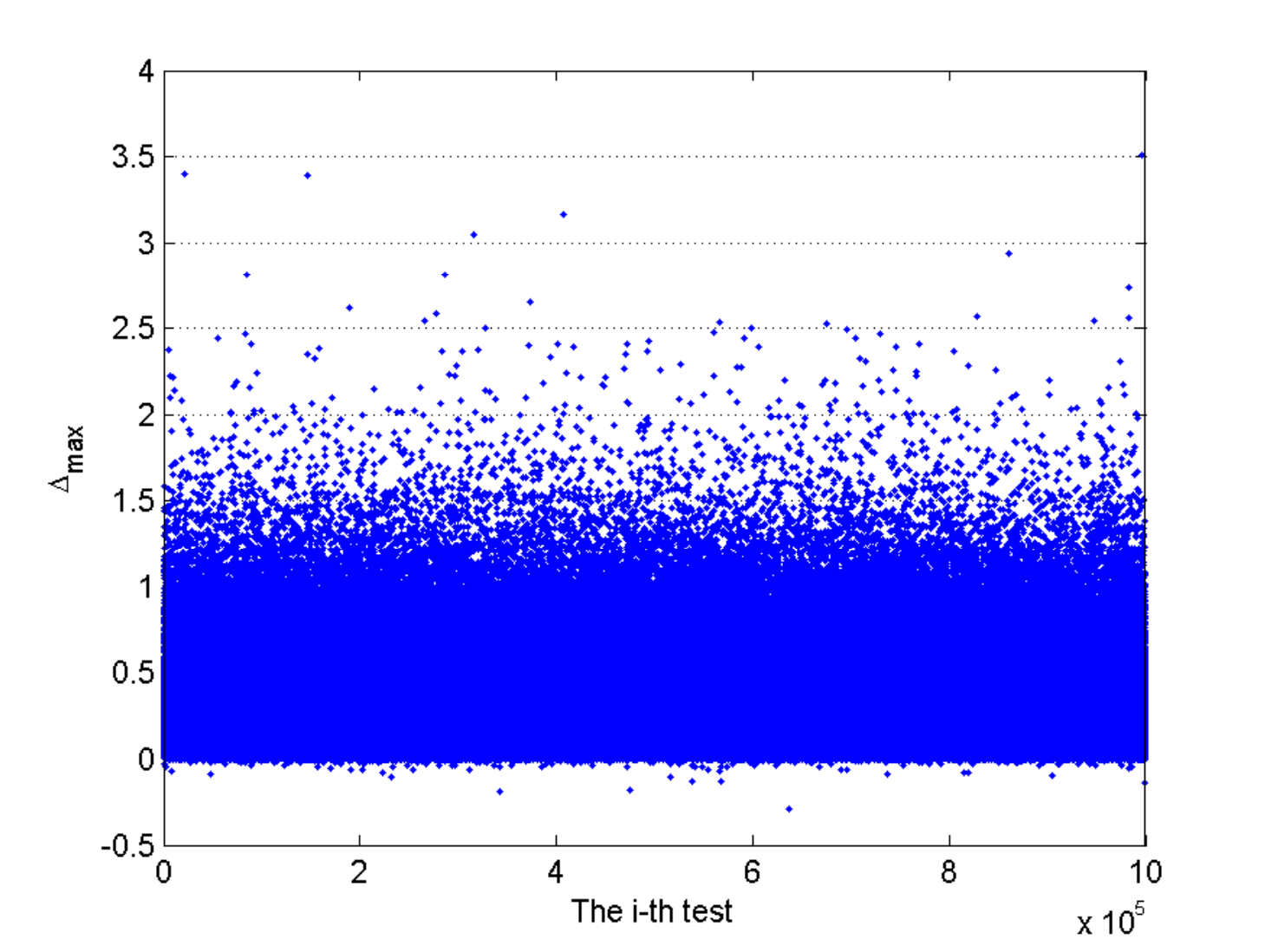}\caption{Two scenarios with respect to one thousand (the left one) and one million (the right one) groups of random data for testing $\Delta_{\max}$.}\label{fig:2}
\end{figure}

\begin{figure}[htbp]
\centering
\includegraphics[width=0.49\textwidth]{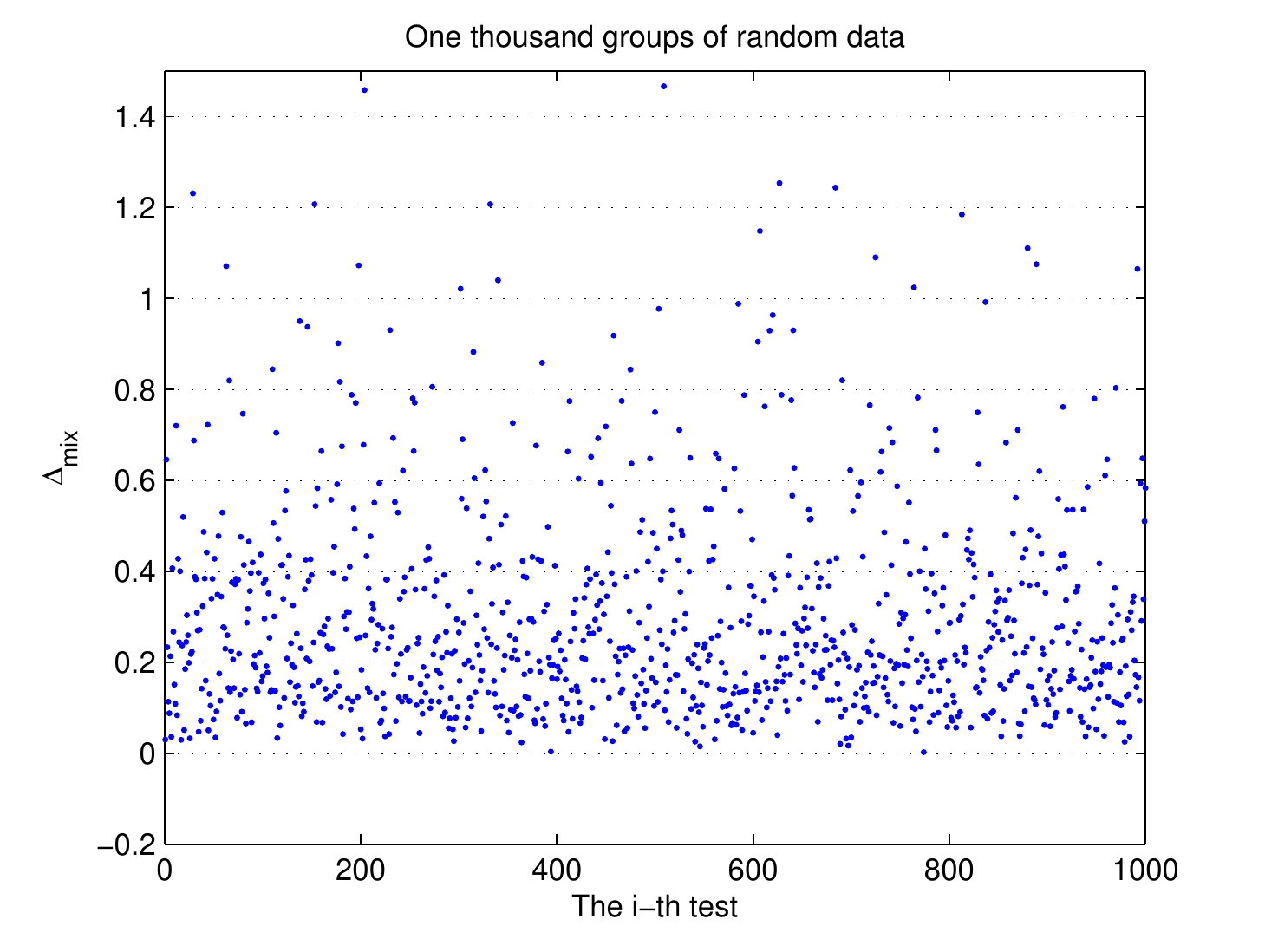}
\includegraphics[width=0.49\textwidth]{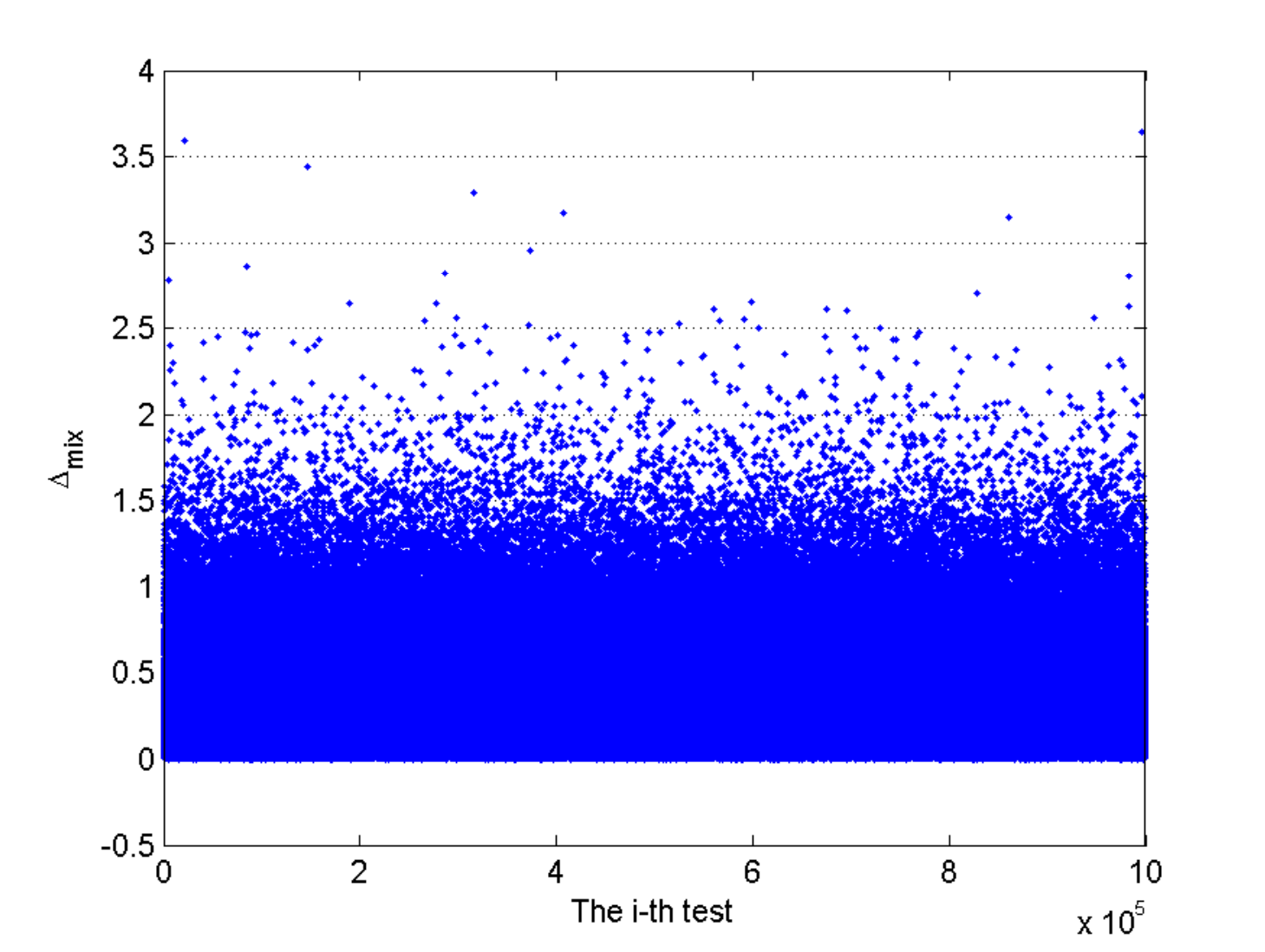}\caption{Two scenarios with respect to one thousand (the left one) and one million (the right one) groups of random data for testing $\Delta_{\text{mix}}$.}\label{fig:3}
\end{figure}

\begin{figure}[htbp]
\centering
\includegraphics[width=0.49\textwidth]{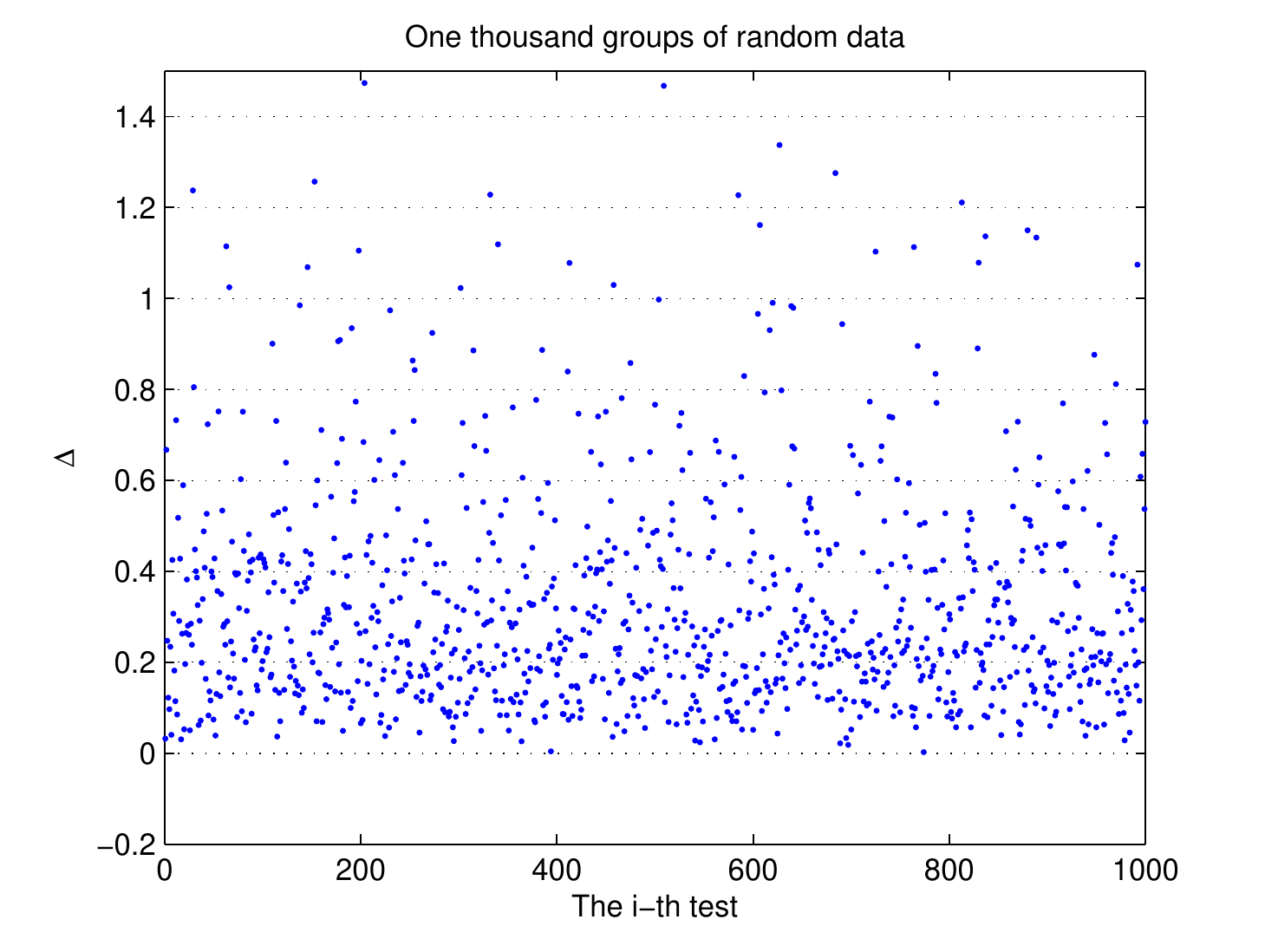}
\includegraphics[width=0.49\textwidth]{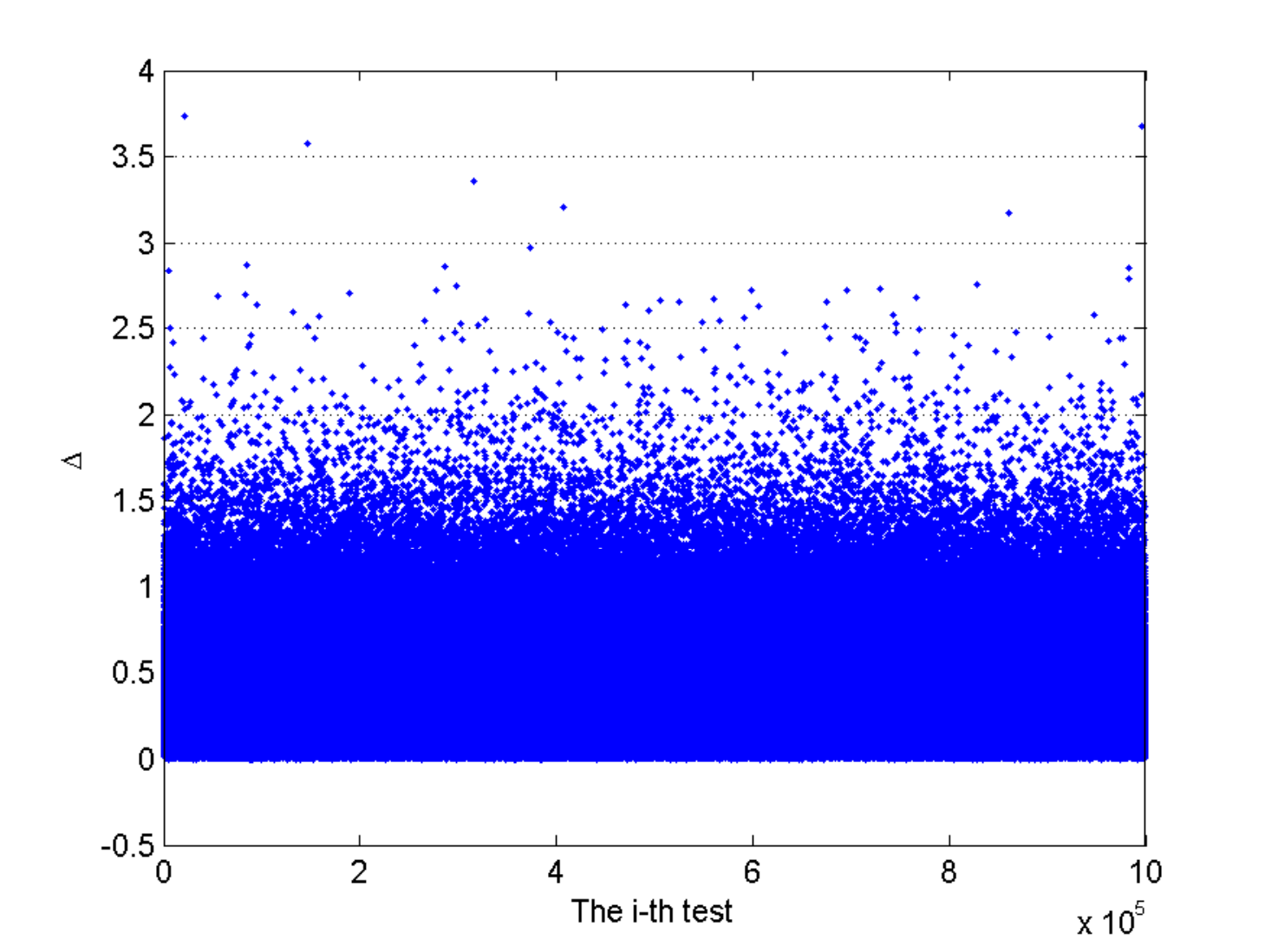}
\caption{Two scenarios with respect to one thousand (the left one) and one million (the right one) groups of random data for testing $\Delta$.}\label{fig:4}
\end{figure}

%%%=============== New version of tex ==================
%\begin{figure}[htbp]
%\centering
%\includegraphics[width=0.49\textwidth]{Delta_min1k-eps-converted-to.pdf}
%\includegraphics[width=0.49\textwidth]{Delta_min_1bw-eps-converted-to.pdf}\caption{Two scenarios with respect to one thousand (the left one) and one hundred thousand (the right one) groups of random data for testing $\Delta_{\min}$.}\label{fig:1}
%\end{figure}
%
%\begin{figure}[htbp]
%\centering
%\includegraphics[width=0.49\textwidth]{Delta_max1k-eps-converted-to.pdf}
%\includegraphics[width=0.49\textwidth]{Delta_max_1bw-eps-converted-to.pdf}\caption{Two scenarios with respect to one thousand (the left one) and one hundred thousand (the right one) groups of random data for testing $\Delta_{\max}$.}\label{fig:2}
%\end{figure}
%
%\begin{figure}[htbp]
%\centering
%\includegraphics[width=0.49\textwidth]{Delta_mix1k-eps-converted-to.pdf}
%\includegraphics[width=0.49\textwidth]{Delta_mix_1bw-eps-converted-to.pdf}\caption{Two scenarios with respect to one thousand (the left one) and one hundred thousand (the right one) groups of random data for testing $\Delta_{\text{mix}}$.}\label{fig:3}
%\end{figure}
%
%\begin{figure}[htbp]
%\centering
%\includegraphics[width=0.49\textwidth]{Delta1k-eps-converted-to.pdf}
%\includegraphics[width=0.49\textwidth]{Delta_1bw-eps-converted-to.pdf}
%\caption{Two scenarios with respect to one thousand (the left one) and one hundred thousand (the right one) groups of random data for testing $\Delta$.}\label{fig:4}
%\end{figure}

%===========================================================================%
\section{Conclusion}
%===========================================================================%

In this context, we conducted numerical studies on the modified
super-additivity of relative entropy. These data strongly support
the following inequality: for qubit pair $(\rho_{AB},\sigma_{AB})$,
\begin{eqnarray}
\rH(\lambda^\downarrow(\rho_{AB})||\lambda^\uparrow(\sigma_{AB})) \geqslant
 \rH(\lambda^\downarrow(\rho_A)||\lambda^\downarrow(\sigma_A)) +
\rH(\lambda^\downarrow(\rho_B)||\lambda^\downarrow(\sigma_B)).
\end{eqnarray}
We guess the conjectured inequality hold for a general qudit pair
$(\rho_{AB},\sigma_{AB})$.

Our numerical studies show that the super-additivity inequality of relative entropy is indeed not valid globally even for full-ranked states:
$$
\rS(\rho_{AB}||\sigma_{AB}) \ngeqslant \rS(\rho_A||\sigma_A) + \rS(\rho_B||\sigma_B).
$$

In the future research, we will consider the following constrained optimization problems under local unitary transformations:
\begin{eqnarray}
&&\max_{U_A\in\unitary{\cH_A},U_B\in\unitary{\cH_B}} \rS(U_A\ot U_B \rho_{AB}U^\dagger_A\ot U^\dagger_B||\sigma_{AB}),\\
&&\min_{U_A\in\unitary{\cH_A},U_B\in\unitary{\cH_B}} \rS(U_A\ot U_B
\rho_{AB}U^\dagger_A\ot U^\dagger_B||\sigma_{AB}).
\end{eqnarray}
Along this line, some investigations has already been done, for
instance, Gharibian in \cite{Gharibian} proposed a measure of
nonclassical correlations in bipartite quantum states based on local
unitary operations; Giampaolo \emph{et. al} in \cite{Giampaolo}
derived the exact relation between the global state change induced
by local unitary evolutions (in particular being generated by a
local Hamiltonian) and the amount of quantum correlations; moreover
they showed that only those composite quantum systems possessing
non-vanishing quantum correlations have the property that any
nontrivial local unitary evolution changes their global state. The
proposed optimization problems are the subject of ongoing
investigations and we hope to report on them in the future.

%===========================================================================%
\subsubsection*{Acknowledgement}

This work is supported by NSFC (Nos.11301123, 11301124, 11361065).

%===========================================================================%

%=============================================================================%

\end{document}